\documentclass[reprint]{JASA}

\usepackage{NotationMacros}

\usepackage{letltxmacro}
\LetLtxMacro{\originaleqref}{\eqref}
\renewcommand{\eqref}{Eq.~\originaleqref}

\begin{document}

\title[Damping Density of a Shoebox Room]{Damping Density of an Absorptive Shoebox Room derived from the Image-Source Method}


\author{Sebastian J. Schlecht}
\email{sebastian.schlecht@aalto.fi}

\altaffiliation{Also at: Department of Art and Media, Aalto University, P.O. Box 13100, Aalto, Finland}

\author{Karolina Prawda}

\affiliation{Acoustics Lab, Department of Information and Communications Engineering, Aalto University, P.O. Box 13100, FI-00076 Aalto, Finland}

\author{Rudolf Rabenstein}

\author{Maximilian Sch\"{a}fer}
\altaffiliation{The work of M. Sch\"{a}fer is funded by the Bavarian State Ministry of Science and the Arts in the framework of the bidt Graduate Center for Postdocs.}

\affiliation{Multimedia Communications and Signal Processing, Digital Communications, Friedrich-Alexander-Universit\"{a}t Erlangen-N\"{u}rnberg (FAU), Germany}

\altaffiliation{All authors have contributed equally.}


%
\begin{abstract}
The image-source method is widely applied to compute room impulse responses (RIRs) of shoebox rooms with arbitrary absorption. However, with increasing RIR lengths, the number of image sources grows rapidly, leading to slow computation. In this paper, we derive a closed-form expression for the damping density, which characterizes the overall multi-slope energy decay. The omnidirectional energy decay over time is directly derived from the damping density. The resulting energy decay model accurately matches the late reverberation simulated via the image-source method.   
The proposed model allows the fast stochastic synthesis of late reverberation by shaping noise with the energy envelope. Simulations of various wall damping coefficients demonstrate the model's accuracy. The proposed model consistently outperforms the energy decay prediction accuracy compared to a state-of-the-art approximation method. The paper elaborates on the proposed damping density's applicability to modeling multi-sloped sound energy decay, predicting reverberation time in non-diffuse sound fields, and fast frequency-dependent RIR synthesis.
\end{abstract} 
\maketitle


\section{Introduction}\label{sec:intro}
Accurately simulating room impulse responses (RIRs) is vital for many applications in room acoustics and signal processing, including tasks like source separation, speech enhancement, dereverberation \cite{Aknin:2021.StochasticReverberationModel}, and architectural acoustics \cite{Savioja:2015ft}. The image-source method (ISM) is frequently employed for modeling RIRs in shoebox rooms \cite{Allen:1979cn, Habets:2006.RIRgen}. The shoebox room is also a standard test case to discuss room acoustical theory \cite{Kuttruff:2009vl}.

Once surpassing a specific time point, termed the \textit{transition time} \cite{Badeau:2019.FrameworkStochastic}, and exceeding a specific frequency known as the \textit{Schroeder frequency} \cite{Kuttruff:2009vl, Badeau:2019.FrameworkStochastic}, RIRs can be modeled stochastically \cite{Aknin:2021.StochasticReverberationModel, Badeau:2019.FrameworkStochastic, Badeau:2018um, Aknin:2020.StochasticImageSource}. Lehmann and Johansson introduced a method to offer a stochastic approximation of the ISM for an absorptive shoebox room \cite{Lehmann:2010kh,Lehmann:2007fi,Lehmann:2008co}. This methodology is especially advantageous during late reverberation phases as the ISM can demand significant computational resources to achieve accurate outcomes.


We have found several choices in the derivations in \cite{Lehmann:2010kh,Lehmann:2007fi,Lehmann:2008co}, which lead to poor estimation for lesser absorptive rooms. We refer to this state-of-the-art method as Lehmann's method. In this work, we derive the late energy decay of a shoebox room based on a stochastic formulation of the ISM. We derive the directional and then direction-integrated damping density. The closed-form expression of the damping density is also compatible with our previously proposed model based on modal approximation \cite{schaefer2023}. 

As an application, we propose a fast method to obtain multi-slope decay characteristics of the late reverberation that shows excellent agreement with the ISM. A comparison with Lehmann's method \cite{Lehmann:2010kh} shows an improved prediction of the energy decay curve. Also, the proposed method does not require an energy matching with the ISM solution, simplifying the synthesis process. As a second application, the damping density can be used to predict the reverberation time (RT). We show the prediction of the damping density in comparison to other RT prediction formulas. For machine learning applications, the closed form of the damping density is amenable for automatic differentiation and can be used to solve inverse acoustic problems similar to \cite{Zhi.2023:DiffISM}.

The paper is organized as follows. Section~\ref{sec:problemDef} introduces the ISM and problem formulation. Section~\ref{sec:method} proposes a novel closed-form expression for the damping density. Section~\ref{sec:results} presents the experimental validation and comparison to the ISM and Lehmann's method. Section~\ref{sec:application} shows applications of the damping density to predict RT and synthesize late reverberation efficiently and accurately. All Matlab code necessary to use the method and reproduce the figures are available online \cite{code}.

\begin{figure*}[!t]
\includegraphics[trim=0cm 0cm 0cm 0cm,width=\textwidth]{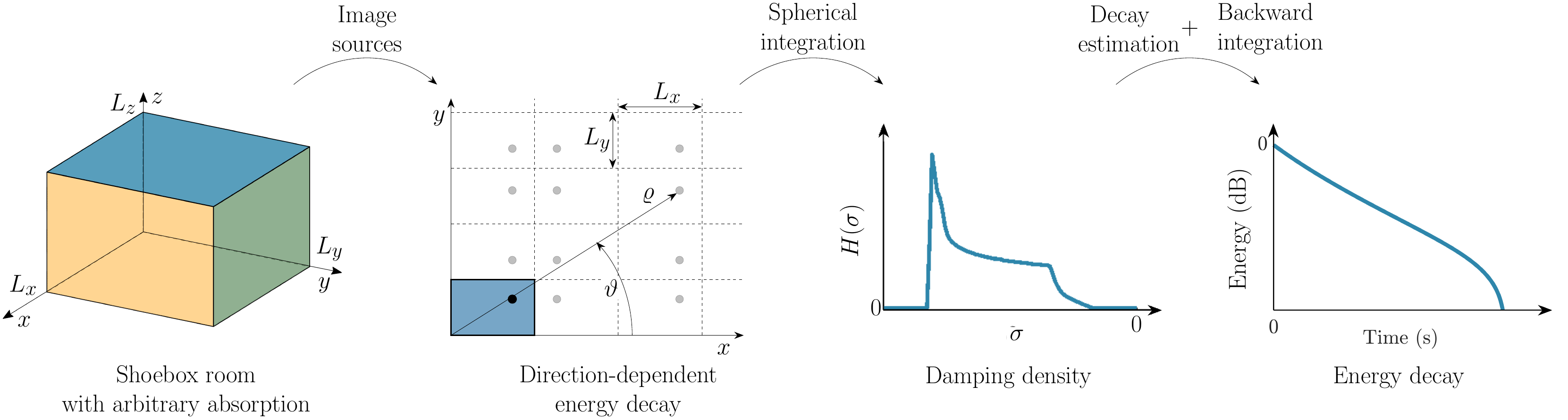}
  \vspace*{-0.6cm}
  \caption{\small Conceptual overview of this paper. From left to right: a shoebox room with dimensions $\Lx, \Ly, \Lz$ and volume $V = \Lx \Ly \Lz$, having non-rigid walls, each characterized by a frequency-independent reflection factor $\r$. We obtain the direction-dependent energy decay from the ISM. From that, we derive the damping density $H(\damping)$ by spherical integration. In the end, the RIR's EDC is estimated from $H(\damping)$ by \eqref{eq:dampingDensity} and backward integration in \eqref{eq:edc}.
  }
  \label{fig:ConceptDrawing}
  \vspace*{-0.4cm}
\end{figure*}

\section{Problem Definition}
\label{sec:problemDef}
 
The coordinate system uses polar angle $\varphi \in[0, \pi]$ and azimuth angle 
$\vartheta \in[0, 2\pi[$. We use a time-dependent propagation distance $\varrho \triangleq \varrho(t)=c t$, where $c$ is the speed of sound in $\si{\meter\per\second}$ and $t$ is time in $\si{\second}$. 

The shoebox room is placed such that the center is at the origin of the coordinate system. The six reflection coefficients are denoted by $\rx, \rxx, \ry, \ryy, \rz, \rzz$. The room size is $\Lx, \Ly, \Lz$ in $\si{\meter}$.

The RIR $h(t)$ generated by ISM \cite{Allen:1979cn} is given by 
\begin{equation}
\begin{aligned}
h(t) &= \sum_{\substack{\vec{m} \in \mathcal{M} \\ \vec{p} \in \mathcal{P}}} \frac{g_{\textrm{ext}}(\dist(t)-d_{\vec{m},\vec{p}})}{4 \pi \dist(t)} \rx^{\abs*{\mx-q}} \rxx^{\abs*{\mx}} \\ 
& \ry^{\abs*{m_y-j}} \ryy^{\abs*{\my}} \rz^{\abs*{\mz-k}} \rzz^{\abs*{\mz}} ,
\label{eq:ISM}	
\end{aligned}
\end{equation}
where $\mathcal{M}=\left\{\left(\mx, \my, \mz\right):-N \leq \mx, \my, \mz \leq N\right\}$, $\mathcal{P}=\left\{\left(q,j,k\right):0 \leq q,j,k \leq 1\right\}$, $g_\mathrm{ext}$ is the excitation term \cite{schaefer2023}, and $d_{\vec{m},\vec{p}}$ is the distance of the image source to the receiver. 

The power response is obtained by squaring $h(t)$ from \eqref{eq:ISM} and short-time averaging as follows
\begin{equation}
\powerResponse(t) = \langle h(t)^2 \rangle.
\label{eq:powerResponse}
\end{equation}

Applying a Schroeder backward integration to the power response, the energy decay curve (EDC) can be obtained as follows 
\begin{equation}
\begin{aligned}
\mathrm{EDC}(t) =\int_t^{\infty} \powerResponse(\tau) \mathrm{~d} \tau.
\label{eq:edc}	
\end{aligned}
\end{equation}
In Fig.~\ref{fig:ConceptDrawing}, the conceptual overview of the proposed method is shown. Instead of a direct derivation from the RIR, as in \eqref{eq:powerResponse}, we first derive the damping density $H(\sigma)$ \cite{Kuttruff:1958uq, schaefer2023} to obtain the power response as follows 
\begin{equation}
\begin{aligned}
\powerResponse(c t) = \powerResponse(\varrho) = \int_{-\infty}^0 H(\sigma) e^{\sigma \varrho} \mathrm{d} \sigma.
\label{eq:dampingDensity}
\end{aligned}
\end{equation} 
The damping density $H(\sigma)$ is, therefore, a distribution of energy decays over time. More precisely, $H(\sigma)$ is the initial energy of the slope, which decays exponentially by $\sigma$. The power response $\powerResponse$ is the superposition of all possible exponential decay. Thus, the damping density $H(\sigma)$ is an insightful representation of the late RIR. In the following, we propose a closed-form expression to compute the damping density from the shoebox-room parameters. 


\section{Proposed Method} \label{sec:method}
In the following, we derive the damping density $H(\sigma)$, where Theorem~\ref{th:sphere} is the main result. The proposed method is related to Lehmann's method, which, however, directly derives the power response $\powerResponse$. We present a few corrections compared to Lehmann's method \cite{Lehmann:2008co}, which are described in the footnotes. First,  in Sec.~\ref{subsec:dirdamp}, the directional damping density is derived from the ISM based on the stochastic approximation of the late reverberation. Second, in Secs.~\ref{subsec:intAz} and \ref{subsec:intPol}, we integrate this directional damping density over the sphere to obtain the omnidirectional damping density. 

To illustrate the concepts and derivations, we use a running example of a shoebox room with $\rx, \rxx, \ry, \ryy, \rz, \rzz = [-1, -1, -3, -2, -2, -5] (\si{\decibel})$  = $[0.891,  0.891,  0.707,  0.794,  0.794,  0.562] $ and room size $\Lx, \Ly, \Lz = [4,5,3] \si{\meter}$.

\subsection{Directional Damping Density}
\label{subsec:dirdamp}
The \textit{directional damping density} characterizes the energy decay of the late reverberation in a specific direction relative to the center of the shoebox room. We derive the directional damping density from a stochastic approximation of the ISM.

From \eqref{eq:ISM}, the expected energy of an image source placed at position $\varrho, \vartheta, \varphi$ can be expressed as follows 
\begin{equation}
\begin{aligned}
P(\varrho, \vartheta, \varphi) = \frac{(\rx\rxx)^{\Wx} (\ry\ryy)^{\Wy} (\rz\rzz)^{\Wz}}{(4 \pi \varrho )^2},
\label{eq:ISpower}	
\end{aligned}
\end{equation}
where $1 / (4\pi\varrho)^2$ is the energetic distance attenuation of a spherical wave \cite{Lehmann:2008co} and the coefficients in the exponents 
\begin{equation}
\begin{aligned}
& \Wx=\frac{\varrho}{\Lx} \abs*{ \cos(\vartheta) \sin(\varphi)}, \\
& \Wy=\frac{\varrho}{\Ly} \abs*{\sin(\vartheta) \sin(\varphi)}, \\
& \Wz=\frac{\varrho}{\Lz} \abs*{\cos(\varphi)},
\end{aligned}
\end{equation}
are the expected number of wall hits in the respective axis\footnote{Compared to \cite{Lehmann:2008co}: Change to spherical coordinates.}.
With the average number of reflections at distance $\varrho$, given by \cite[Eq.~(4.5)]{Kuttruff:2009vl}\footnote{Compared to \cite{Lehmann:2008co}:  $\varrho^2$ instead of $\varrho$}
\begin{equation}
\begin{aligned}
\refldens(\varrho) = 4 \pi \frac{\varrho^2}{V},
\label{eq:reflectionDensity}	
\end{aligned}
\end{equation}
the power response $\powerResponse$ results from integrating the image source energy $P$ over the sphere\footnote{Compared to  \cite{Lehmann:2008co}: Added $\sin(\varphi)$ for correct integration over the sphere}
\begin{equation}
\powerResponse(\varrho)=\refldens(\varrho) \int_0^{\pi} \int_0^{2\pi} P(\varrho, \vartheta, \varphi) \sin(\varphi) \mathrm{d} \vartheta \mathrm{d} \varphi.
\label{eq:powerResponseSphere}
\end{equation}
We define the directional damping density $Q$ relative to the omnidirectional damping density by  
\begin{equation}
\begin{aligned}
H(\sigma) = \int_0^{\pi} \int_0^{2\pi} Q(\sigma, \vartheta, \varphi) \sin(\varphi) \dd{\vartheta} \dd{\varphi}.
\label{eq:directionalToOmniDensity}	
\end{aligned}
\end{equation}
By inserting \eqref{eq:directionalToOmniDensity} into \eqref{eq:dampingDensity}, comparing it to \eqref{eq:powerResponseSphere}, and by swapping the integration order, we can obtain the following expression of the directional damping density
\begin{equation}
\begin{aligned}
\refldens(\varrho) P(\varrho, \vartheta, \varphi) = \int_{-\infty}^0 Q(\sigma, \vartheta, \varphi) e^{\sigma \varrho} \mathrm{d} \sigma .
\label{eq:directionalDampingDensity}	
\end{aligned}
\end{equation}
The left-hand side of \eqref{eq:directionalDampingDensity} can also be expressed as 
\begin{equation}
    \refldens(\varrho) P(\varrho, \vartheta, \varphi) 
    = \frac{1}{4 \pi V} \exp(\varrho M(\vartheta, \varphi)),
    \label{eq:DP2}
\end{equation}
with the function 
\begin{equation}
\begin{aligned}
    M(\vartheta, \varphi) &= \Kx \abs*{\cos(\vartheta) \sin(\varphi)} \nonumber\\&+ \Ky \abs*{\sin(\vartheta) \sin(\varphi)} + \Kz \abs*{\cos(\varphi)},
    \label{eq:functionM}
\end{aligned}
\end{equation}
and the coefficients $K$ in $\si{\per\meter}$, depending on the reflection coefficients and room dimensions 
\begin{equation}
\begin{aligned}
& \Kx= \ln \p*{\rx \rxx} / \Lx, \\
& \Ky= \ln \p*{\ry \ryy}/ \Ly, \\
& \Kz= \ln\p*{\rz \rzz}/ \Lz.
\label{eq:KxKyKz}
\end{aligned}
\end{equation}
Exploiting \eqref{eq:DP2}, we get the right-hand side of \eqref{eq:directionalDampingDensity} to obtain an expression for the directional damping density
\begin{equation}
\begin{aligned}
Q(\sigma, \vartheta, \varphi) = \frac{1}{4 \pi V} \delta(\sigma - M(\vartheta, \varphi)),
\label{eq:directionalDampingDiracs}	
\end{aligned}
\end{equation}
where $\delta(\cdot)$ is the Dirac delta function. Inspecting \eqref{eq:directionalDampingDiracs}, it can be seen that the power response in each direction is always a single slope, so multi-slope decay only emerges by integrating over different directions.

To obtain the omnidirectional damping density $H(\sigma)$ in \eqref{eq:directionalToOmniDensity}, the directional damping density $Q(\sigma,\theta,\varphi)$ in \eqref{eq:directionalDampingDiracs} needs to be integrated over the sphere. 
In the following, we solve the integration in \eqref{eq:directionalToOmniDensity} in two steps. First, we integrate over the azimuth angle $\vartheta$, and second, we integrate over the polar angle $\varphi$.



\begin{figure}[!t]
  \includegraphics[]{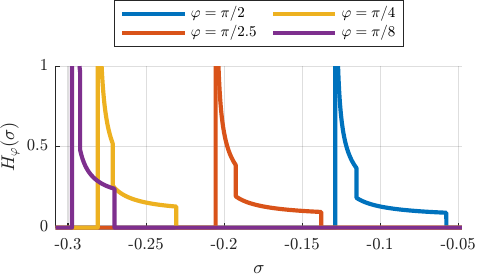}
  \caption{Horizontal slice of the damping density $H_{\varphi}(\sigma)$ in \eqref{eq:2DdampingDensity} for the example room at four different polar angles $\varphi$.}
  \label{fig:dampingDensity2D}
\end{figure}

\subsection{Integration over Azimuth Angle}
\label{subsec:intAz}
First, we integrate over the azimuth angle to obtain the damping density $H_\varphi(\sigma)$ for a horizontal slice by choosing a fixed $\varphi$. Due to the symmetry of the functions $M(\vartheta, \varphi)$ in every octant of the sphere, we conveniently choose the octant with $\varphi, \vartheta \in [0, \pi/2]$ such that $\cos(\vartheta) \sin(\varphi)$, $\sin(\vartheta) \sin(\varphi)$, $\cos(\varphi) \geq 0$.
At a fixed $\varphi$, function $M$ in \eqref{eq:functionM} is then 
\begin{equation}
\begin{aligned}
M(\vartheta, \varphi)\big\vert_{\varphi\in[0,\pi/2]} \!\!= M_{\varphi}(\vartheta) = 
-\sqrt{\KxSin^2 + \KySin^2}\cos\left(\vartheta + \phi\right) + \KzCos,
\label{eq:2D}
\end{aligned}
\end{equation}
with the variables 
\begin{equation}
\begin{aligned}
\KxSin &=  \Kx \sin(\varphi), &\KySin &= \Ky \sin(\varphi), \\
\KzCos &= \Kz \cos(\varphi), &\phi &= \arctan(-\KySin/\KxSin) = \arctan(-\Ky / \Kx). 
\label{eq:alphaBetaGamma}	
\end{aligned}
\end{equation}
For the chosen octant, $\varphi, \vartheta \in [0, \pi/2]$, we have $\KxSin$, $\KySin$, $\KzCos$, $\phi$ $\leq 0$. The values $\KxSin$, $\KySin$, and $\KzCos$ can be considered Cartesian coordinates of points on an ellipsoid. 

Fixing the angle $\varphi$ and inserting \eqref{eq:2D} into \eqref{eq:directionalDampingDiracs} simplifies the integral in \eqref{eq:directionalToOmniDensity} to a relation for the damping density $H_\varphi$ as follows
\begin{equation}
\begin{aligned}
H_{\varphi}(\sigma) &= 8 \int_0^{\pi / 2}  Q(\sigma, \vartheta, \varphi)  \mathrm{d} \vartheta \\
&= \frac{8}{4 \pi V} \int_0^{\pi / 2}  \delta(\sigma - M_{\varphi}(\vartheta))  \mathrm{d} \vartheta .
\label{eq:2DdirectionalToOmniDensity}	
\end{aligned}
\end{equation}
The factor of 8 is included to compensate for integrating over only one octant.

\begin{lemma} \label{th:horizontal}
The damping density $H_\varphi(\sigma)$ for a horizontal slice with fixed angle $\varphi\in[0,\,2\pi]$ is given by 
\begin{equation}
\begin{aligned}
H_{\varphi}(\sigma) = \frac{8}{4 \pi V} \frac{\mu_\varphi(\sigma - \KzCos)}{\sqrt{\KxSin^2 + \KySin^2 - (\sigma - \KzCos)^2}},
\label{eq:2DdampingDensity}	
\end{aligned}
\end{equation}
where $\mu_\varphi$ is an indicator function as defined in \eqref{eq:indicator}.
\end{lemma}

\begin{proof}
Please refer to Appendix~\ref{sec:appAzimuth}. 
\end{proof}


Figure~\ref{fig:dampingDensity2D} shows the damping density $H_{\varphi}(\sigma)$ for different values of $\varphi$ for the example room. For $\varphi = \pi / 2$, we have $\KxSin = -0.057$, $\KySin = -0.115$, $\KzCos = 0$, and $-\sqrt{\KxSin^2 + \KySin^2} = -0.129$. The support of the distribution is between $\KxSin$ and $-\sqrt{\KxSin^2 + \KySin^2}$. Note the discontinuity at $\KySin$ is due to the step in the indicator function $\mu_\varphi$. When next integrating over the polar angle, this discontinuity will lead to multiple piecewise sections of the damping density.

\subsection{Integration over Polar Angle}
\label{subsec:intPol}
To derive the omnidirectional damping density $H(\sigma)$, we need to integrate $H_\varphi$ over the polar angle $\varphi$, i.e., 
\begin{equation}
\begin{aligned}
&H(\sigma) = \int_0^{\pi / 2} H_{\varphi}(\sigma) \sin(\varphi) \mathrm{d} \varphi.
\label{eq:3DdirectionalToOmniDensity2}
\end{aligned}
\end{equation}

The following theorem states the closed-form expression of the damping density, which is the main result of this work.
\begin{theorem} \label{th:sphere}
The omnidirectional damping density is
\begin{equation}
\begin{aligned}
H(\sigma) &= \frac{8}{4 \pi V} \p*{2 H_0(\sigma) - H_1(\sigma) - H_2(\sigma)} ,\\
H_i(\sigma) &= F(\sigma,u) \at_{\cos \hat{\varphi}_i^+}^{\cos \hat{\varphi}_i^-},
\label{eq:dampingDensityCloseForm}	
\end{aligned}
\end{equation}
with 
\begin{equation}
\begin{aligned}
F(\sigma,u) &= - \frac{1}{\sqrt{-A}} \arcsin(\frac{2A u + B}{\Delta}), \\
A &= - \Kx^2 - \Ky^2 - \Kz^2, \quad
B = 2 \sigma \Kz, \\
C &= \Kx^2 + \Ky^2 - \sigma^2, \quad
\Delta = \sqrt{B^2 - 4AC},
\label{eq:this1}	
\end{aligned}
\end{equation}
and evaluation limits
\begin{equation}
\begin{aligned}
\hat{\varphi}_i^+ &= \max(\varphi_i^+,0), \quad
\hat{\varphi}_i^- = \min(\varphi_i^-,\pi/2), \\
\varphi_{i} &= \pm \arccos(-\frac{\sigma}{\sqrt{a_i^2 + \Kz^2}}) + \arctan( \frac{a_i}{\Kz}), \\
a_0 &= -\sqrt{\Kx^2 + \Ky^2}, \,
a_1 = \Kx, \, 
a_2 = \Ky . 
\label{eq:this2}	
\end{aligned}
\end{equation}
\end{theorem}

\begin{proof}
Please refer to Appendix~\ref{sec:appPolar}. 
\end{proof}

Figure~\ref{fig:dampingDensity3D} shows the omnidirectional damping density for the example room. Also, the seven special points, where the curve is non-differentiable, are visible: 
\begin{equation}
\begin{aligned}
\Kx = -0.057, \Ky = -0.115, \Kz &= -0.268, \\ -\sqrt{\Kx^2 + \Ky^2} &= -0.128, \\ -\sqrt{\Kx^2 + \Kz^2} &= -0.274, \\ -\sqrt{\Ky^2 + \Kz^2} &= -0.292, \\ -\sqrt{\Kx^2 + \Ky^2 + \Kz^2} &= -0.297 \,.
\label{eq:specialPoints}	
\end{aligned}
\end{equation}
Thus, the damping density is distributed between $-\sqrt{\Kx^2 + \Ky^2 + \Kz^2}$ and $ \max(\Kx,\Ky,\Kz)$. This concludes the derivation of the damping density. In the following sections, we validate the derived expression and provide application examples.


\begin{figure}[!t]
  \includegraphics[]{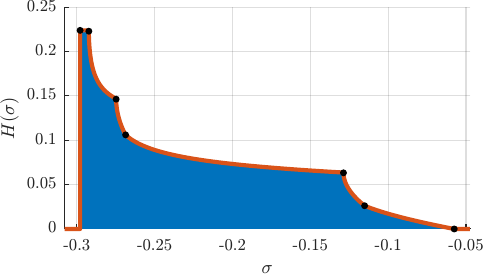}
  \caption{The omnidirectional damping density for the example room. The blue bars are the histogram computed from \eqref{eq:directionalToOmniDensity}. The red line is the closed-form expression from \eqref{eq:dampingDensityCloseForm}. The black dots indicate the seven special points; see \eqref{eq:specialPoints}.}
  \label{fig:dampingDensity3D}
\end{figure}

\begin{figure}[!t]
  \includegraphics[]{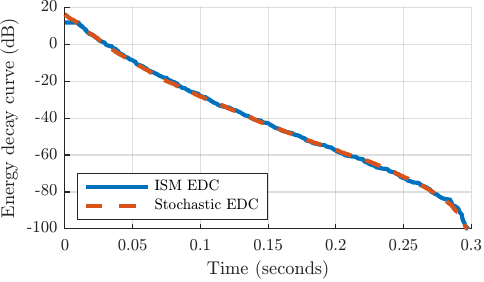}
  \caption{Energy decay curves of the example room with derived power response, \eqref{eq:numericalPowerResponse}, compared to the  ISM, \eqref{eq:ISM}.}
  \label{fig:EDCs} 
\end{figure}




\section{Experimental Validation}\label{sec:results}

In the following, we validate the proposed damping density by comparing it to the results of ISM and Lehmann's method. We evaluate the accuracy of the tested approaches using RT differences from the ground-truth ISM estimation in typical and extreme scenarios.

\subsection{Shaped Decaying Noise}
\label{subsec:omni}
We use the power response resulting from the damping density to synthesize an RIR by shaping noise. The power response is computed by solving the integral in \eqref{eq:dampingDensity}. We did not find an analytical expression and performed instead a numerical integration. The damping density varies slowly, so even coarse sampling of the damping coefficients gives accurate results.

For simplicity, we choose a uniform set of damping coefficients $\sigma_i$ between $ -\sqrt{\Kx^2 + \Ky^2 + \Kz^2}$ and $\max(\Kx, \Ky, \Kz)$. Then, we evaluate the damping density $H(\sigma_i)$ for this set using \eqref{eq:dampingDensityCloseForm}.
Finally, the power response can be computed then by a weighted sum (replacing the integration in \eqref{eq:dampingDensity}) as follows
\begin{equation}
\begin{aligned}
    \powerResponse(\varrho) = \sum_i H(\sigma_i) e^{\sigma_i \varrho},
\label{eq:numericalPowerResponse}	
\end{aligned}
\end{equation}
for target distance or time range $\varrho = c t$. The sum in \eqref{eq:numericalPowerResponse} corresponds to the real axis of the discrete Laplace transform for which fast algorithms are available \cite{Loh:2023}.   

The power response can now be used to compute an approximation of the RIR.
Stochastic late reverberation synthesis is performed by applying the power response envelope to stationary noise, i.e.,

\begin{equation}
\begin{aligned}
h_{\textrm{stochastic}}(t) = u(t) \sqrt{\powerResponse(c t)},
\label{eq:stochasticRIR}	
\end{aligned}
\end{equation}
where $u(t)$ is a normalized Gaussian white noise.

The EDCs of the proposed stochastic RIR and ISM for the example room are compared in Fig.~\ref{fig:EDCs}. The excellent match confirms that the proposed stochastic response models the late part of ISM RIR accurately.

\begin{figure}[!t]
  \includegraphics[]{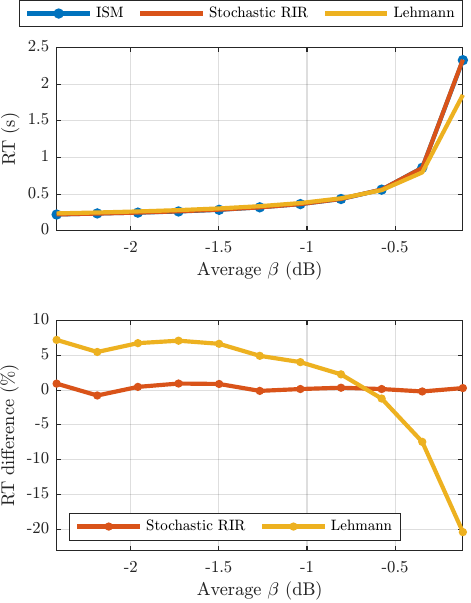}
  \caption{Comparison of RTs obtained with the three methods over increasing reflection coefficient values. (Top) obtained RT values and (bottom) the relative difference between ISM and the proposed and Lehmann's method.}
  \label{fig:RTplots}
\end{figure}

\subsection{Energy Decay Validation}
We further evaluate the EDCs synthesized with the proposed method against the ground-truth decay obtained with ISM. Additionally, we compare the accuracy of decay reproduction with Lehmann's method \cite{Lehmann:2008co, Lehmann:2007fi}. 

We synthesize a broad range of reflection coefficients to account for various decay rates. We start from $\rx, \rxx, \ry, \ryy, \rz, \rzz = [-0.161,~-0.180,~-0.025,~-0.181,~-0.125,~-0.018] (\si{\decibel})$ and increase the $\si{\decibel}$ value of reflection coefficients by a factor of $1,3,5,\dots,21$, changing the average $\r$ from $-0.115$ to $-2.305$~$\si{\decibel}$ in $11$ steps.

Unlike in the proposed method, the source and receiver positions influence the RIRs synthesized with ISM. This might create differences in EDCs that will give exaggeratedly different results. To minimize such an effect, we keep one set of reflection coefficients in each trial and the room geometry fixed at $L =[4,5,3]$\,m, while the source and receiver placement in ISM are determined by random sampling within the room. The positions are drawn from a uniform distribution over $3$ iterations. RIRs obtained in each trial are then averaged to obtain a single EDC.  

\begin{figure}[!t]
  \includegraphics[]{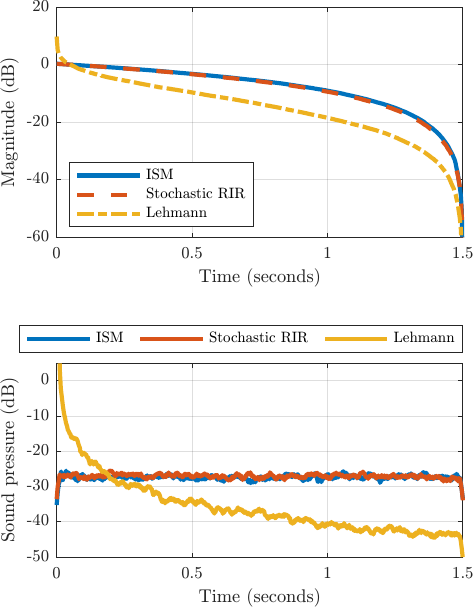}
  \caption{The ``lossless" case with rigid walls. Top: Energy decay curves. Bottom: Short-time average of sound pressure in dB with window lengths of 21~ms.}
  \label{fig:losslessEDC}
\end{figure}

Fig.~\ref{fig:RTplots} shows the evaluation results regarding the RT. The top pane presents the RT values for the three compared methods, while the bottom pane depicts the relative difference between ISM and either the proposed model or Lehmann's method. Lehmann's method displays a bias when the average reflection coefficient is below $-0.6$~dB, always showing RT values $5-10\%$ higher than the ISM. The bias is reversed for high reflection coefficients, where Lehmann's method underestimates the RT by 20\%. 
For our proposed method, the maximum relative error is 1.45\%, where the median absolute relative error is $0.56\%$.
Therefore, it is shown that the proposed stochastic approach outperforms Lehmann's method, resulting in lower error across the whole range of considered reflection coefficients. 

To reveal a weak point of Lehmann's method, we set all the reflection coefficients to $-0.0001$ (dB) to achieve a ``lossless" case with a very slow energy decay. The top pane of Fig.~\ref{fig:losslessEDC} shows the resulting EDCs for the three evaluated methods, while the bottom pane of Fig.~\ref{fig:losslessEDC} illustrates the respective sound pressure plots. The ISM and the proposed stochastic method are in excellent agreement and behave identically, displaying a slowly decaying EDC and a flat sound pressure level over time. Lehmann's method, however, reveals an EDC with a more pronounced decay while the sound pressure level drops by around $15$\,dB over $1.5$\,s.
Thus, the proposed method clearly improves reverberation estimation for rooms with very low absorption. 


\section{Applications of Damping Density}
\label{sec:application}
In this section, we showcase different applications of damping density to illustrate its practical value as an intermediate parameterization of stochastic reverberation. We present examples of multi-slopedness, RT prediction, and fast RIR synthesis. 

\subsection{Multi-Slopeness of Energy Decay}

In his classic paper, Kuttruff hypothesized possible shapes of the damping density to analyze the shape of the energy decay \cite[Fig.~1]{Kuttruff:1958uq}. A damping density with a single peak leads to a single exponential decay, which is the ubiquitous assumption made in reverberation time estimation. A damping density with two peaks is referred to as double-sloped and often occurs in coupled rooms \cite{Luizard.2014} and reverberation chambers \cite[see Fig.~1(A)]{Balint:2019.Bayesian}. A recently proposed common-slope model suggests that in complex geometry, source-receiver positions only impact the peak heights but not the peak position \cite[]{Goetz:2023.CommonSlope}. Fig.~\ref{fig:dampingDensity3D} illustrates the more general shape of the damping density, which is reminiscent of curves hypothesized by Kuttruff \cite[see curves $b$ and $c$ in Figure 3.12]{Kuttruff:2009vl}.

The analytical damping density resulting from ISM has an inherently more complicated distribution. The shape of the damping density is determined by three values: $\Kx$, $\Ky$, and $\Kz$ in \eqref{eq:KxKyKz} and an overall scaling factor. The seven special points in \eqref{eq:specialPoints}  determine the support of the damping density. Thus, even for rooms of the same material and room dimensions, i.e., $\Kx = \Ky = \Kz$, the damping density is distributed between $\sqrt{3} \Kx$ and $\Kx$. This suggests that the ISM of a shoebox room cannot result in a single exponential decay. 

\begin{figure}[!t]
  \includegraphics[]{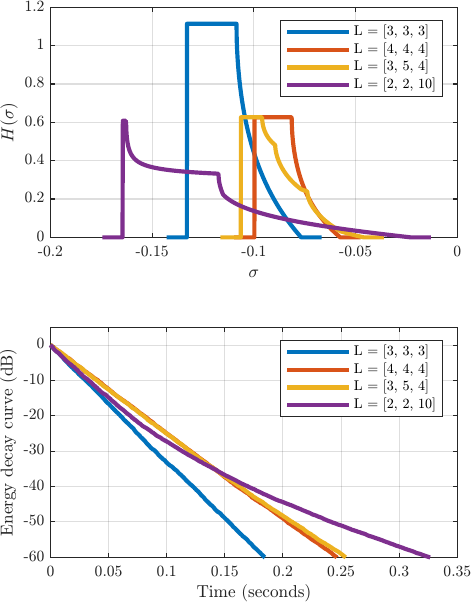}
  \caption{(Top) Damping density for different room sizes with equal reflection coefficients $\r_i = -1$~dB for all walls. (Bottom) Corresponding EDCs.}
  \label{fig:DensityVariation}
\end{figure}

The top pane of Fig.~\ref{fig:DensityVariation} shows the damping density using \eqref{eq:dampingDensityCloseForm} for different room sizes when the reflection coefficients are equal, whilst the bottom pane depicts the corresponding EDCs. The cubic room $L = [3,3,3]$ (blue line) has the most compact support and regular shape, leading to an almost linear EDC. The larger cubic room $L = [4,4,4]$ (red line) has the same shape but overall less damping and lower energy, suggesting a longer overall decay rate than the smaller cubic room. The example room dimension $L = [3, 5, 4]$ (yellow line) has wider support but similar average damping to the larger cubic room. The decay rates for both rooms are alike, but the more irregular shape of the damping density for the room with $L = [3, 5, 4]$ results in a more non-linear behavior of the corresponding EDC after 50~dB of decay. The support gets wider for a corridor-shaped room with $L = [2,2,10]$ (purple line), indicating a larger variety of decay rates than in the remaining rooms. The respective EDC deviates the most from the linear decay.

\subsection{Comparison to Classic Reverberation Time Formulas}

There are several classic formulas to predict the reverberation time from the geometry and reflection coefficients, e.g.,
\begin{equation}
\begin{aligned}
T_{60, \mathrm{Sab}} &=\frac{0.161 \cdot V}{\sum_{i=1}^6 S_i \alpha_i}, \\
T_{60, \mathrm{Eyr}} &=\frac{0.161 \cdot V}{-S \cdot \log \left(1-\sum_{i=1}^6 S_i \alpha_i / S\right)}, \\
T_{60, \mathrm{Fit}} &=\frac{0.161 \cdot V}{S^2} \cdot\left(\frac{-2 \Ly \Lz}{\log \left(1-\left(\alpha_{x0}+\alpha_{x1}\right) / 2\right)}-\right. \\
&\left.\frac{2 \Lx \Lz}{\log \left(1-\left(\alpha_{y0}+\alpha_{y1}\right) / 2\right)}-\frac{2 \Lx \Ly}{\log \left(1-\left(\alpha_{z0}+\alpha_{z1}\right) / 2\right)}\right),
\label{eq:T60s}	
\end{aligned}
\end{equation}
where $\alpha_{kj} = 1 - \r_{kj}^2$ with $k = x,y,z$ and $j = 0,1$, are the absorption coefficients, $S$ is the total surface area of the enclosure, and $S_i, i  = 1, \ldots, 6 $, are the surface areas of each individual wall.

The proposed damping density can also be used to predict the reverberation time in shoebox rooms. While this expression is not as simple as classic formulas by Sabine or Eyring for RT prediction, it can still be evaluated with a few computations. Additionally, the closed form of $H(\sigma)$ in \eqref{eq:dampingDensity} potentially allows calibration of the parameters based on measurements \cite{Prawda:2022.Calibrating}.

As previously discussed, the energy decay is multi-sloped such that the definition of the RT is not unique. In particular, the measured dynamic range of the EDC impacts the predicted slope. For instance, $T_{20}$, which is assessed on a $20$-$\si{\decibel}$ evaluation range, may lead to lower RT values than $T_{60}$, which is based on the $60$-$\si{\decibel}$ range.

\begin{figure}[!t]
  \includegraphics[]{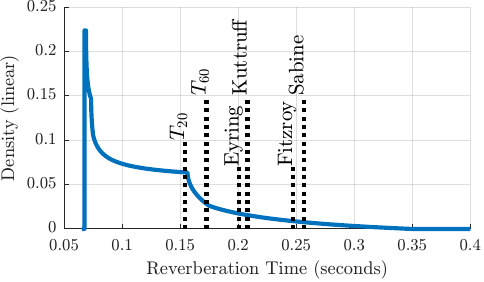}
  \caption{RT density of the example room and corresponding RT predictions. The $T_{20}$ and $T_{60}$ are derived from the EDC of the proposed method.}
  \label{fig:RTdensity}
\end{figure}

Fig.~\ref{fig:RTdensity} shows the RT density of the example room, indicating different RT predictions, including the proposed $T_{20}$ and $T_{60}$. The damping constant $\sigma$ is converted to RT values by
\begin{equation}
\begin{aligned}
T_{60,\sigma} = \frac{-60}{20\log_{10}(e^{\sigma c})} = \frac{-6.9078}{\sigma c}~.
\label{eq:sigmaRT}	
\end{aligned}
\end{equation}

\begin{figure}[!t]
  \includegraphics[]{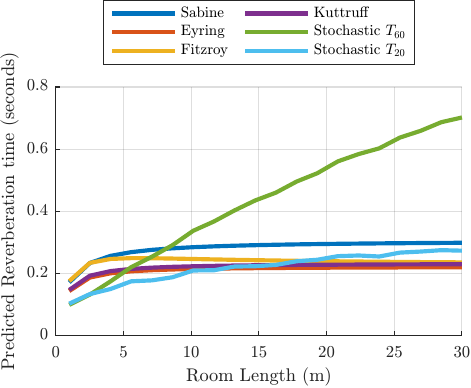}
  \caption{Predicted RT comparing classic formulas against prediction based on ISM. All parameters are the same as in the example room. Only the first room dimension (length) is altered.}
  \label{fig:RTCorridor}
\end{figure}

Fig.~\ref{fig:RTCorridor} shows the RT predictions when changing the room lengths. The parameters are like in the example room but with room size $L = [l,5,3] $ meters with $l = 1, \dots, 30$. The growing disproportion between the room dimensions leads to damping density shapes resembling that of the corridor-shaped room in Fig.~\ref{fig:DensityVariation}. The resulting multi-slopedness leads to non-linear EDC and  diverging estimates between $T_{20}$ and $T_{60}$. Fig.~\ref{fig:RTCorridor} shows that the classic formulas saturate for larger room sizes, while the proposed stochastic method increases more linearly.

\begin{figure}[!th]
  \includegraphics[]{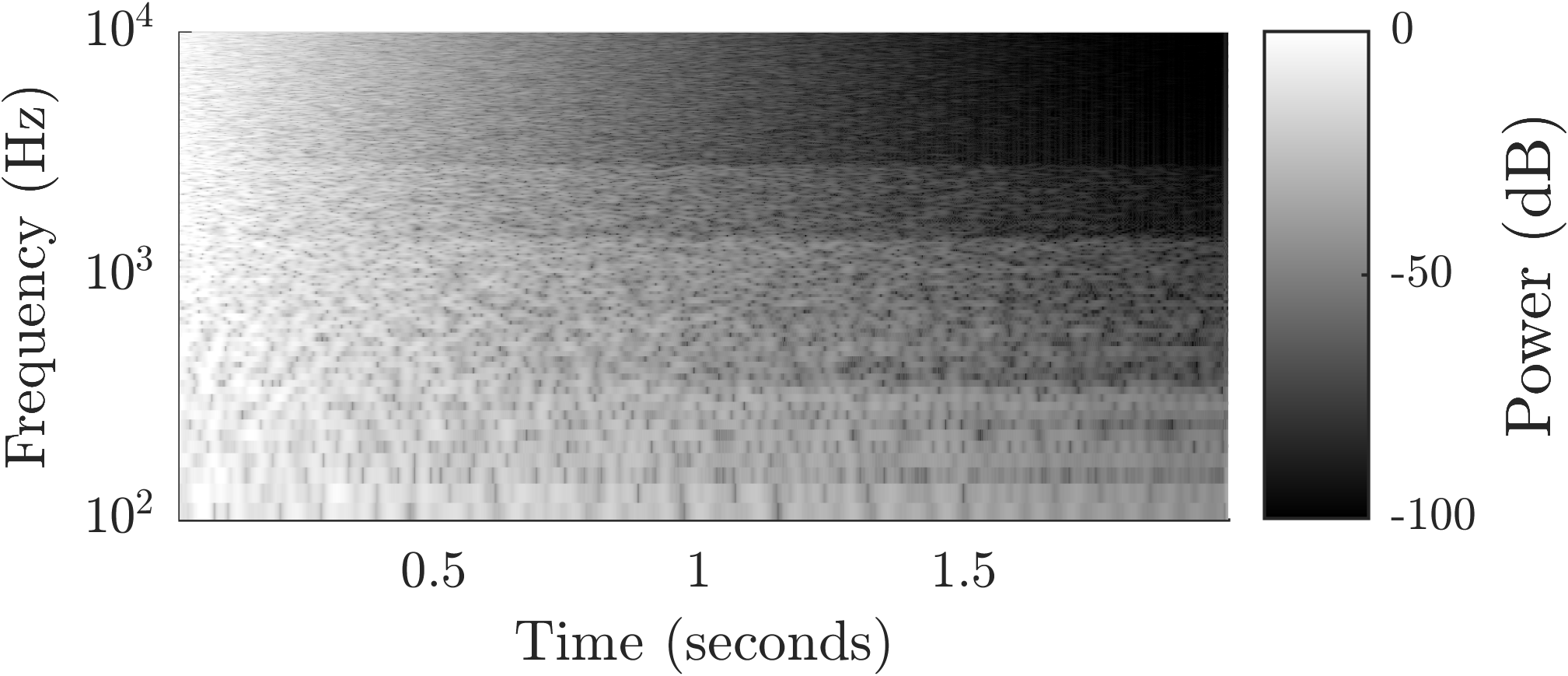}
  \includegraphics[]{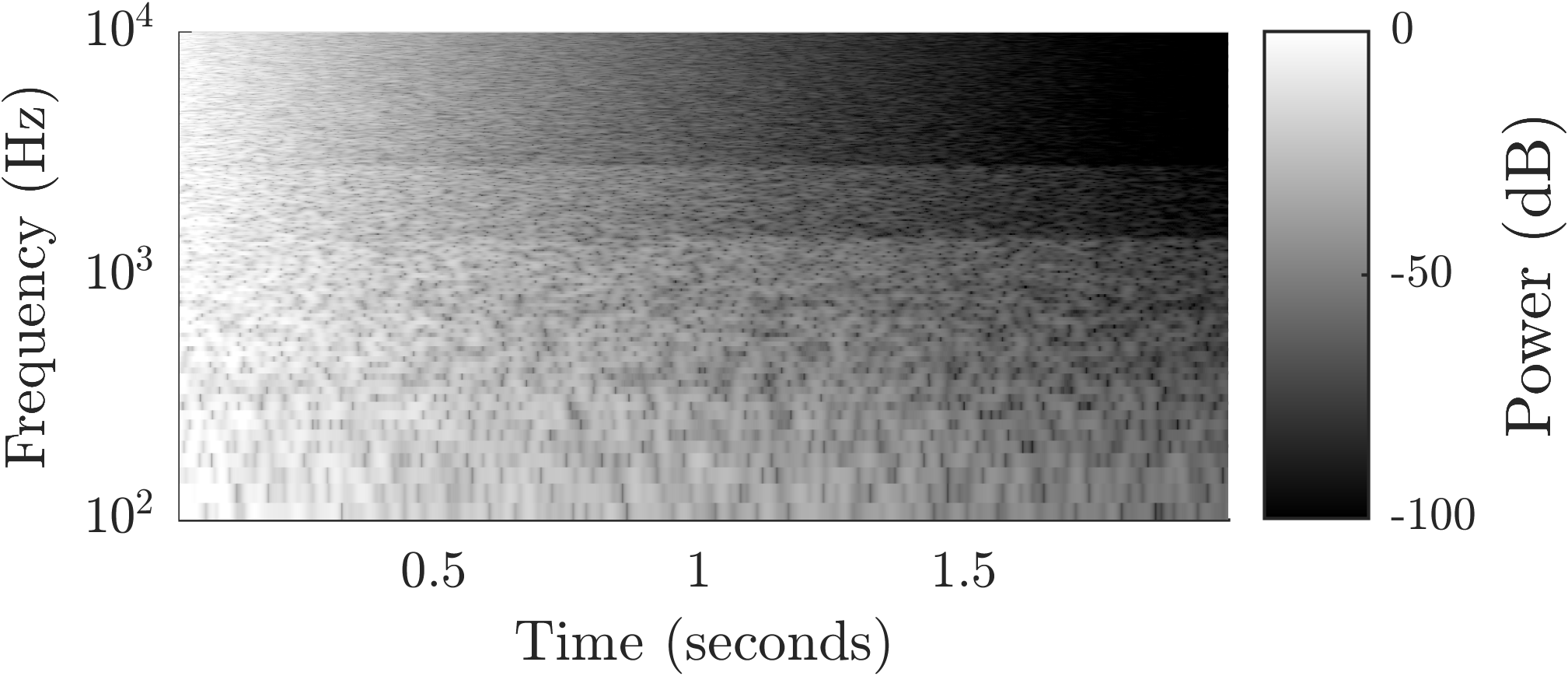}
  \caption{Time-frequency energy decay for frequency-dependent reflection coefficients. Top: ISM simulation. Bottom: Stochastic RIR.}
  \label{fig:spectrograms}
\end{figure}

This example illustrates that ISM and the classic formulas make fundamentally different assumptions. While the classic formulas are based on diffuse reflections, the ISM is strictly specular. Which prediction applies more to a given room depends strongly on these model assumptions. Our recent study suggests that the classic formulas can be calibrated so that even a shoebox room conforms to Sabine's and Eyring's predictions \cite{Prawda:2022.Calibrating}.


\subsection{Fast RIR Synthesis}
The late part of the reverberation can often be replaced with a stochastic approximation without impairing the perceptual quality. Based on some predetermined transition time, the early part from ISM is combined with the shaped noise. This procedure readily applies to frequency-dependent reflection coefficients, leading to frequency-dependent damping density. The shaped noise synthesis can then be applied per frequency band.


We show here an example of reflection coefficients specified per octave band. For each octave band in the range of 125\,Hz to 4\,kHz, we derive the reflection coefficients based on the target RT$= [2, 1.6, 1.4, 1.2, 1, 0.8]$s and the dimension of the example room, $L = [4,5,3]$m. The damping density is computed for the respective reflection coefficients in each octave band. The resulting envelope shapes a broadband Gaussian noise and is filtered through an octave bandpass filter. We obtain the final stochastic RIR  by summing the bandpassed signals. 

Figure~\ref{fig:spectrograms} shows spectrograms of the time-frequency energy decay resulting from the ISM method and the proposed stochastic RIR. In contrast to the earlier method \cite{Lehmann:2010kh}, no additional level matching is necessary as all scaling factors are already applied. In the low frequency of the ISM, some modal resonances are visible, which are not reproduced by the stochastic RIR synthesis based on shaped noise. The low-frequency similarity can be improved with a modal reverberator derived from the ISM  \cite{schaefer2023}.  The corresponding sound examples are provided online \cite{code}.

\section{Conclusion}
\label{sec:conclusion}

In this paper, we proposed a method to describe the late part of the room impulse response using a stochastic approach. We derived the omnidirectional damping density of an absorptive shoebox room compatible with the energy decay resulting from the ISM. The spherical integration of the directional damping density leads to a characteristic distribution showing that ISM consistently exhibits a multi-slope decay. The overall energy is derived rigorously, which allows synthesizing a stochastic RIR without energy matching to be easily combined with an early part based on the ISM. 

The evaluation shows that the synthesized reverberation is more accurately matched with the ground-truth ISM than another established method. The proposed model allows for accurate modeling of multi-slope decays and reverberation time prediction under the assumption of a non-diffuse soundfield. Future work may include the weighting of directional sources and receivers. The damping density can be also derived for more complex room geometries and can serve as a useful  representation of energy decay.

\section{Author Declarations}
The authors have no conflicts of interest to disclose.

\section{Data Availability}
The data and code that support the findings of this study are openly available in a GitHub repository \cite{code}.


\appendix
\section{Derivation of density $H_\varphi$}\label{sec:appAzimuth}

To obtain the damping density $H_\varphi(\sigma)$ we start with \eqref{eq:2DdirectionalToOmniDensity}, i.e., 
\begin{equation}
    H_\varphi(\sigma) = \frac{8}{4 \pi V} \int_0^{\pi / 2}  \delta(\sigma - M_{\varphi}(\vartheta))  \mathrm{d} \vartheta, 
    \label{eq:appHphi}
\end{equation}
where the function $M_\varphi$ is defined in \eqref{eq:2D}. Next, we want to exploit the following property of the Dirac delta function~\cite{Rabenstein_Schaefer_MDSS_23}
\begin{equation}
\begin{aligned}
\delta(f(\vartheta)) = \sum_i \frac{1}{\abs{ \pdv{}{\vartheta} f(\vartheta_i)}} \delta(\vartheta - \vartheta_i),
\label{eq:appDelta1}	
\end{aligned}
\end{equation}
where $\vartheta_i$ is a simple root of the function $f(\vartheta)$. Therefore, inspecting \eqref{eq:appHphi} we define the the function $f_\sigma(\vartheta)$ and its derivative as follows
\begin{equation}
\begin{aligned}
f_\sigma(\vartheta) &= \sigma - M_{\varphi}(\vartheta) \\
&=  \sigma + \sqrt{\KxSin^2 + \KySin^2} \cos(\vartheta + \phi) - \KzCos, \\
\pdv{}{\vartheta} f_\sigma(\vartheta) &= -\sqrt{\KxSin^2 + \KySin^2} \sin(\vartheta + \phi) .
\label{eq:appFsig}	
\end{aligned}
\end{equation}
Using \eqref{eq:trigRoots}, the two roots $\vartheta_i$ of $f_\sigma(\vartheta)$ can be obtained as 
\begin{equation}
\begin{aligned}
    \vartheta_{1,2} = \pm \arccos(\frac{\sigma - \KzCos}{-\sqrt{\KxSin^2 + \KySin^2}}) - \phi.
\label{eq:appRootFsig}	
\end{aligned}
\end{equation}

\noindent By inserting the roots $\vartheta_{1,2}$ from \eqref{eq:appRootFsig} into the derivative of $f_\sigma$ in \eqref{eq:appFsig} yields 
\begin{equation}
\begin{aligned}
\pdv{}{\vartheta} f_\sigma(\vartheta_{1,2})
&= -\sqrt{\KxSin^2 + \KySin^2} \sin(\pm\arccos{\frac{\sigma - \KzCos}{-\sqrt{\KxSin^2 + \KySin^2}}}) \\ 
&= \mp\sqrt{\KxSin^2 + \KySin^2} \sqrt{1 - \frac{(\sigma - \KzCos)^2}{\KxSin^2 + \KySin^2}} , \\ 
\abs{\pdv{}{\vartheta} f_\sigma(\vartheta_{1,2})} 
&= \sqrt{\KxSin^2 + \KySin^2 - (\sigma - \KzCos)^2}.
\label{eq:appHphiTerm1}	
\end{aligned}
\end{equation}
In particular, the factor is equal for both roots, i.e., 
\begin{equation}
\begin{aligned}
\abs{\pdv{}{\vartheta} f_\sigma(\vartheta_1)} = \abs{\pdv{}{\vartheta} f_\sigma(\vartheta_2)}.
\label{eq:app2DequalDerivates}	
\end{aligned}
\end{equation}
Exploiting the properties \eqref{eq:appDelta1}--\eqref{eq:app2DequalDerivates}, the expression for the density $H_\varphi$ in \eqref{eq:appHphi} can be simplified as follows
\begin{equation}
\begin{aligned}
H_{\varphi}(\sigma) = \frac{8}{4 \pi V}  \sum_{i = 1}^2 \frac{1}{\abs{\frac{\partial}{\partial \vartheta} f_\sigma(\vartheta_i)}} \int_0^{\pi / 2} \delta(\vartheta - \vartheta_i) \partial \vartheta \\
= \frac{8}{4 \pi V}  \frac{1}{\abs{\frac{\partial}{\partial \vartheta} f_\sigma(\vartheta_1)}} \sum_{i = 1}^2 \int_0^{\pi / 2} \delta(\vartheta - \vartheta_i) \partial \vartheta.
\label{eq:appHphiSimp}	
\end{aligned}
\end{equation}

In order to find a closed form expression for $H_\varphi$, we solve the integral in \eqref{eq:appHphiSimp}. We use the following property of the Delta impulse for the integral
\begin{equation}
\begin{aligned}
\int_0^{\pi / 2} \delta(\vartheta - \vartheta_i) \mathrm{d} \vartheta = \1{0,\pi/2}(\vartheta_i).
\label{eq:appHphiTerm2-delta}	
\end{aligned}
\end{equation}
with the indicator function
\begin{equation}
\begin{aligned}
\1{a,b}(x) = \begin{cases}
1 \quad \textrm{ if } a \leq x \leq b \\
0 \quad \textrm{ otherwise .}
\end{cases}.
\label{eq:indicator_funct}	
\end{aligned}
\end{equation}
The condition $0 \leq \vartheta_i \leq \pi/2$ is equivalent to (see \eqref{eq:appRootFsig})
\begin{equation}
\begin{aligned}
 \phi &\leq \pm \arccos(\frac{\sigma - \KzCos}{-\sqrt{\KxSin^2 + \KySin^2}}) \leq \phi + \pi/2  .
 \label{eq:arccos-range}	
\end{aligned}
\end{equation}
Since $-\pi/2 \leq \phi \leq 0$ (see~\eqref{eq:alphaBetaGamma}), the condition \eqref{eq:arccos-range} has to be checked in the range 
$[-\pi/2, \pi/2]$. The $\arccos$-function is removed by applying the function $\sin \phi$ which is monotonically increasing in $[-\pi/2, \pi/2]$
\begin{equation}
\begin{aligned}
\sin \phi &\leq \sin(\pm \arccos(\frac{\sigma - \KzCos}{-\sqrt{\KxSin^2 + \KySin^2}})) \leq \cos \phi .\\
\label{eq:this6}	
\end{aligned}
\end{equation}
We substitute $\phi$ with \eqref{eq:alphaBetaGamma} and apply the identities \eqref{eq:sinatan}, \eqref{eq:cosatan}, \eqref{eq:sinacos}, such that
\begin{equation}
\begin{aligned}
\frac{\KySin}{\sqrt{\KxSin^2 + \KySin^2}} &\leq \pm \sqrt{1 - \frac{(\sigma - \KzCos)^2}{\KxSin^2 + \KySin^2}} \leq \frac{-\KxSin}{\sqrt{\KxSin^2 + \KySin^2}} , \\
\KySin &\leq \pm \sqrt{\KxSin^2 + \KySin^2 - (\sigma - \KzCos)^2} \leq -\KxSin .\\
\label{eq:this7}	
\end{aligned}
\end{equation}
By inspecting the two cases ($\pm$), and for $\sigma < 0$ we see that 
\begin{equation}
\begin{aligned}
0 \leq \vartheta_1 \leq \pi/2 \textrm{ if } -\sqrt{\KxSin^2 + \KySin^2} \leq \sigma - \KzCos \leq \KySin, \\  
0 \leq \vartheta_2 \leq \pi/2 \textrm{ if } -\sqrt{\KxSin^2 + \KySin^2} \leq \sigma - \KzCos \leq \KxSin. 
\label{eq:indicatorBounds}	
\end{aligned}
\end{equation}
With these bounds, the integral in \eqref{eq:appHphiSimp} can be expressed by the following indicator function 
\begin{equation}
\begin{aligned}
\mu_\varphi(\sigma- \KzCos) &= \sum_{i=1}^2 \int_0^{\pi / 2} \delta(\vartheta - \vartheta_i) \mathrm{d} \vartheta\\ 
&= \1{-\sqrt{\KxSin^2 + \KySin^2},\KySin}(\sigma- \KzCos) + \1{-\sqrt{\KxSin^2 + \KySin^2},\KxSin}(\sigma- \KzCos) .
\label{eq:indicator}	
\end{aligned}
\end{equation}
Finally, inserting the results for both terms \eqref{eq:appHphiTerm1} and \eqref{eq:indicator} into \eqref{eq:appHphiSimp}, leads to the closed form expression for $H_\varphi(\vartheta)$ in \eqref{eq:2DdampingDensity}. This concludes the proof.

\begin{figure}[!t]
  \includegraphics[]{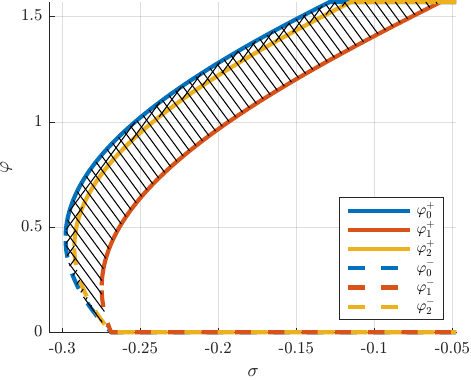}
  \caption{Integration limits for the example room. The cross-hatched areas belong to $H_0(\sigma) - H_1(\sigma)$ and $H_0(\sigma) - H_2(\sigma)$, respectively. The dashed lines are the lower limit, and the solid lines are the upper limit.  }
  \label{fig:integrationLimits}
\end{figure}

\section{Derivation of Density $H$}
\label{sec:appPolar}
We expand the damping density using Lemma~\ref{th:horizontal}
\begin{equation}
\begin{aligned}
&H(\sigma) = \int_0^{\pi / 2} H_{\varphi}(\sigma) \sin(\varphi) \mathrm{d} \varphi \\
&= \frac{8}{4 \pi V} \int_0^{\pi / 2} 
\frac{\sin(\varphi) \mu_{\varphi}(\sigma - \KzCos)}{\sqrt{ \KxSin^2 + \KySin^2 - (\sigma - \KzCos)^2}}  \dd{\varphi}.
\label{eq:3DdirectionalToOmniDensity}
\end{aligned}
\end{equation}

The indicator function $\1{a,b}(x)$ can be expressed as integration limits $x_0 = \max(a,0)$ and $x_1 = \min(b,\pi/2)$, so we first solve the integral with general limits $\varphi_0$ and $\varphi_1$:
\begin{equation}
\begin{aligned}
\int_{\varphi_0}^{\varphi_1} \frac{\sin(\varphi)}{\sqrt{\sin(\varphi)^2 (\Kx^2 + \Ky^2) - (\sigma - \Kz \cos(\varphi))^2}}  \dd{\varphi} = I.
\label{eq:indefiniteIntegral}	
\end{aligned}
\end{equation}
Substituting with $u = \cos \varphi$ and $\dd{u} = - \sin \varphi \dd{\varphi}$, we have
\begin{equation}
\begin{aligned}
\int_{\varphi_0}^{\varphi_1} \sin(\varphi)  \dd{\varphi} = - \int_{\cos \varphi_0}^{\cos \varphi_1}  \dd{u} = \int_{\cos \varphi_1}^{\cos \varphi_0} \dd{u}.
\label{eq:this8}	
\end{aligned}
\end{equation}
We use
\begin{equation}
\begin{aligned}
\KxSin^2 + \KySin^2 - (\sigma - \KzCos)^2 = A \cos^2 \varphi + B \cos \varphi + C,
\label{eq:this9}	
\end{aligned}
\end{equation}
where
\begin{equation}
\begin{aligned}
A &= - \Kx^2 - \Ky^2 - \Kz^2 ,\\
B &= 2 \sigma \Kz ,\\
C &= \Kx^2 + \Ky^2 - \sigma^2 ,\\
\Delta &= \sqrt{B^2 - 4AC} .
\label{eq:this10}	
\end{aligned}
\end{equation}
Thus,
\begin{equation}
\begin{aligned}
I &= \int_{\cos \varphi_1}^{\cos \varphi_0} \frac{\dd{u}}{\sqrt{A u^2 + B u + C}} \\
&= - \frac{1}{\sqrt{-A}} \arcsin(\frac{2A u + B}{\Delta}) \at_{\cos \varphi_1}^{\cos \varphi_0} \\
&= F(\sigma,u) \at_{\cos \varphi_1}^{\cos \varphi_0} .
\label{eq:this11}	
\end{aligned}
\end{equation}
The evaluation limits are the limits of the indicator functions in \eqref{eq:indicator}, i.e.,
\begin{equation}
\begin{aligned}
\sigma - \KzCos &= -\sqrt{\KxSin^2 + \KySin^2} &\rightarrow u_0 = \cos \varphi_0 ,\\
\sigma - \KzCos &= \KySin &\rightarrow u_1 = \cos \varphi_1 ,\\
\sigma - \KzCos &= \KxSin &\rightarrow u_2 = \cos \varphi_2 .
\label{eq:indicator_limits}	
\end{aligned}
\end{equation}
The limits in~\eqref{eq:indicator_limits} are of the generic form
\begin{equation}
\begin{aligned}
a_i \sin(\varphi) + b \cos(\varphi) - c = 0
\label{eq:this12}	
\end{aligned}
\end{equation}
with
\begin{equation}
\begin{aligned}
a_0 &= -\sqrt{\Kx^2 + \Ky^2}, \,
a_1 = \Kx, \, 
a_2 = \Ky, \\
b &= \Kz, \quad
c = \sigma .
\label{eq:this13}	
\end{aligned}
\end{equation}
Using \eqref{eq:trigRoots}, we have with $b \leq 0$
\begin{equation}
\begin{aligned}
\varphi_{i}^\pm = \pm \arccos(-\frac{c}{\sqrt{a_i^2 + b^2}}) - \arctan(- a_i / b) .
\label{eq:indicator_phi}	
\end{aligned}
\end{equation}
The limits need to be between $[0, \pi/2]$. The indicator function can also result in two distinct intervals. To avoid handling multiple cases, we express
\begin{equation}
\begin{aligned}
\ind_{[a,b] \cap [0,\pi/2]}(x) = \ind_{[b_1,b_2] \cap [0,\pi/2]} - \ind_{[a_1,a_2] \cap [0,\pi/2]} .
\label{eq:this14}	
\end{aligned}
\end{equation}
Thus, the evaluation limits are (with $\varphi_{i}^\pm $ from~\eqref{eq:indicator_phi})
\begin{equation}
\begin{aligned}
\hat{\varphi}_i^+ = \max(\varphi_i^+,0), \quad
\hat{\varphi}_i^- = \min(\varphi_i^-,\pi/2),
\label{eq:this15}	
\end{aligned}
\end{equation}
The evaluation limits can be represented graphically in Fig.~\ref{fig:integrationLimits}.

We conclude the proof by writing the damping density as 
\begin{equation}
\begin{aligned}
H(\sigma) &= \frac{8}{4 \pi V} \p*{2 H_0(\sigma) - H_1(\sigma) - H_2(\sigma)} , \\
H_i(\sigma) &= F(\sigma,u) \at_{\cos \hat{\varphi}_i^+}^{\cos \hat{\varphi}_i^-} .
\label{eq:this16}	
\end{aligned}
\end{equation}


\section{Useful Identities}
We list several well-known identities to ease the derivations presented in this paper
\begin{align}
\sin(\arctan(x)) &= \frac{x}{\sqrt{1 + x^2}}, \label{eq:sinatan} \\
\cos(\arctan(x)) &= \frac{1}{\sqrt{1 + x^2}}, \label{eq:cosatan} \\
\sin(\pm \arccos(x)) &= \pm \sqrt{1 - x^2}, \label{eq:sinacos} \\
\arctan(x) &= - \arccos \frac{1}{\sqrt{1 + x^2}} \textrm{ for } x \leq 0,\\
a \sin(x) + b \cos(x) &= \sgn(b) \sqrt{a^2 + b^2} \cos(x + \arctan(- a/b)), \\
\arccos(x) \pm \arccos(y) &= \arccos(xy \mp \sqrt{(1 - x^2)(1 - y^2)}), \\
\arcsin(x) - \arcsin(y) &= \arcsin(x\sqrt{1 - y^2} - y\sqrt{1 - x^2}).
\label{eq:this17}	
\end{align}

The roots of $a \sin(x) + b \cos(x) - c = 0$ for $\sgn(a) = \sgn(b)$ are 
\begin{equation}
\begin{aligned}
x_{1,2} &= \pm \arccos(\frac{c}{\sgn(b)\sqrt{a^2 + b^2}}) - \arctan(-a / b) \\
&=  \pm \arccos(\frac{c}{\sgn(b)\sqrt{a^2 + b^2}}) + \arccos(\frac{b}{\sqrt{a^2 + b^2}}) \\
&= \arccos(\frac{1}{a^2 + b^2} \p*{\frac{c b}{\sgn(b)} \mp a \sqrt{a^2 + b^2 - c^2}} ) .
\label{eq:trigRoots}	
\end{aligned}
\end{equation}

\bibliography{My_Library, JASA_variable}
















\end{document}